\documentclass[pra,10pt,twocolumn,superscriptaddress,floatfix,showpacs]{revtex4-1}

\usepackage[utf8]{inputenc}
\usepackage[T1]{fontenc}     
\usepackage[british]{babel}  
\usepackage[sc,osf]{mathpazo}
\usepackage[scaled=0.86]{berasans}  
\usepackage[colorlinks=true, citecolor=blue, urlcolor=blue]{hyperref}  
\usepackage{graphicx} 
\usepackage[babel]{microtype}  
\usepackage{amsmath,amssymb,amsthm,bm,amsfonts,mathrsfs,bbm} 

\usepackage{xspace}  
\usepackage{pgfplots}
\usepackage{xcolor,colortbl}
\def\ba{\begin{equation}}
\def\ea{\end{equation}}
\def\bea{\begin{eqnarray}}
\def\eea{\end{eqnarray}}
\def\ben{\begin{equation*}}
\def\een{\end{equation*}}
\def\bean{\begin{eqnarray*}}
\def\eean{\end{eqnarray*}}
\def\bma{\begin{mathletters}}
\def\ema{\end{mathletters}}
\def\bi{\begin{itemize}}
\def\ei{\end{itemize}}

\newcommand{\be}{\begin{equation}}
\newcommand{\ee}{\end{equation}}

\newcommand{\kommentar}[1]{}

\newcommand{\forget}[1]{}

\newtheorem{theorem}{Theorem}

\newtheorem{corollary}{Corollary}[theorem]

\begin{document}

\title{Shareability of Quantum Steering and its Relation with Entanglement}

\author{Biswajit Paul}
\email{biswajitpaul4@gmail.com}
\affiliation{Department of Mathematics, Balagarh Bijoy Krishna Mahavidyalaya, Hooghly, West Bengal, India}

\author{Kaushiki Mukherjee}
\affiliation{Department of Mathematics, Government Girls' General Degree College, Ekbalpore, Kolkata, India}


\begin{abstract}
Steerability is a characteristic of quantum correlations lying in between entanglement and Bell nonlocality. Understanding how these steering correlations can be shared between different parties has profound applications in ensuring security of quantum communication protocols. Here we show that at most two bipartite reduced states of a three qubit state can violate the three settings CJWR linear steering inequality contrary to two settings linear steering inequality. This result explains that quantum steering correlations have limited shareability properties and can sometimes even be non-monogamous. In contrast to the two setting measurement scenario, three setting scenario turns out to be more useful to develop deeper understanding of shareability of tripartite steering correlations. Apart from distribution of steering correlations, several relations between reduced bipartite steering, different measures of bipartite entanglement of reduced states and genuine tripartite entanglement are presented here. The results enable detection of different kind of tripartite entanglement.

\end{abstract}


\maketitle

	
\section{Introduction}

Success of a secure quantum network depends on quantum correlations distributed and shared among different parties over many sites \cite{Kim}. Different kind of quantum correlations, for instance, multipartite entanglement\cite{Hor,Guh} and multipartite nonlocality \cite{Bru} have been extensively used as a resource to perform many task in such networks. A key property of these quantum correlations used to secure those quantum networks is that they have limited shareability properties and sometimes can even be monogamous. For example, when a quantum system $A$ is entangled with another system $B$ then this entanglement puts a constraint on the amount of entanglement that can exist between one of those parties ($B$, say) and a third party, $C$. This limited shareability phenomenon is termed as monogamy. This is one of the fundamental differences between quantum entanglement and classical correlations, where all classical probability distributions can be shared \cite{Ton}. Monogamy of entanglement was first quantified by Coffman, Kundu, and Wootters (CKW) in \cite{Cof}, where it was shown that the sum of the individual pairwise entanglement between A and B and C cannot exceed the entanglement between A and the remaining parties together. Since then many research work have been done on such monogamy relations of quantum entanglement \cite{Osb,Yco,Bai,Zhu,Reg,Teh}. This characteristic of quantum entanglement has found potential applications in various quantum information processing tasks such as quantum key distribution \cite{Pal,Gis}, classification of quantum states \cite{Dur,Gio,Pra}, study of black-hole physics \cite{Llo}, and frustrated spin systems \cite{Rao}, etc. Similar to monogamy of entanglement, if any two quantum systems $A$ and $B$ are correlated in such a way that they violate Bell-CHSH inequality \cite{Cla} then neither of $A$ nor $B$ can be Bell-CHSH nonlocal with the other system $C$. In the last few years, several fundamental results on shareability of nonolocal correlations have been proven that constrain the distribution of nonlocal correlations in terms of violation of some Bell-type inequalities among the subsystems of a multipartite system\cite{Ton,Sca,Scn,Tne,Bar,Mas,See,Paw,Kur,Qin,Che,Tra,Ram} and they play a key role in the applications of quantum nonlocal correlations to cryptography\cite{Pal,Gis}. Monogamy relations have also been studied for quantum discord \cite{Ste}, indistinguishability \cite{Kar}, coherence \cite{Rad} , measurement induced nonlocality \cite{Chn} and other nonclassical correlations\cite{Chn}.\\

Despite the importance of shareability in quantum information, the knowledge of shareability for quantum steering is so far rather limited \cite{Rei,Xia,Chg}. The objective of this paper is to understand more about the shareability associated with the quantum steering. The notion of steering was introduced by Schr$\mathrm{o}$dinger in 1935  \cite {Sch} and the effect was recently formalised from foundational as well as quantum information perspective \cite{Wis,Jon}. Considering two distant observers Alice and Bob sharing an entangled state, steering captures the fact that Alice, by performing a local measurement on her subsystem, can remotely steer Bob's state. This is not possible if the shared state is only classically correlated. This kind of quantum correlation is known as steering \cite{Uol}.
It can be understood as a form of quantum nonlocality intermediate between entanglement and Bell nonlocality \cite{Qui}. Quantum steering is certified by the violation of steering inequalities. A number of steering inequalities have been designed to observe steering \cite{Red,Cav,Wal,Sco,Jlc,Kog,Cva,Zuk,Jev,Cos}.
Violation of such steering inequalities certify the presence of entanglement in one-sided device-independent way. Steerable states were shown to be beneficial for tasks involving randomness generation \cite{Law}, subchannel discrimination
\cite{Pia}, quantum information processing \cite{Bra}, and one-sided device-independent processing in quantum key distribution \cite{Ban}. However, comparatively little is known about the shareability of this type of nonlocality. By deriving shareability relations, one can understand how this special type of nonlocal correlation (steering) can be distributed over different subsystems. In this paper, by using the three settings CJWR linear steering inequality \cite{Cav,Cos}, we will derive different kind of trade-off relations that quantify the amount of bipartite steering that can be shared among the three qubit systems. In turn, these trade-off relations help us to prove that at most two of three reduced states of an arbitrary three qubit state can violate the three settings CJWR linear steering inequality contrary to two settings CJWR linear steering inequality or Bell-CHSH inequality, where at most one of the reduced states can violate those inequalities. Consequently, in general, steering correlations turn out to be non-monogamous.\\
Over the past few years it has become clear that correlation statistics of two-body subsystems can be very fruitful in inferring
the multipartite properties of a composite quantum system \cite{Bne,Wni,Tur,Tot,Kor,Kna,Las}. In this context, we have also studied how the reduced bipartite steering of a three qubit state depends on the bipartite and genuine tripartite entanglement of the three-qubit states. Interestingly, criteria for detecting different kind of entanglement of pure three qubit state are obtained based on these shareability relations. We illustrate the relevance of our results with different examples.

\section{PRELIMINARIES}\label{sec2}
In this section, we briefly discuss the concept of steering and the three settings CJWR linear steering inequality that we use in this work.
\subsection{Steering}
 Steering is usually formulated by considering a quantum information task \cite{Wis,Jon}. Suppose two spatially separated observers, say Alice and Bob share a bipartite state $\rho_{AB}$ and they can perform measurements in the sets $M_{A}$ and $M_{B}$,
respectively. In a steering test, Bob, who does trust  his own but not Alice's apparatus, wants to verify whether the shared state between them is entangled. He will be convinced that the shared state $\rho_{AB}$ is entangled  only if his system is genuinely influenced by Alice's measurement, instead of some preexisting local hidden states (LHS) which Alice may have access to. To make sure that Bob must  exclude the LHS model
 \begin{equation}\label{mon1}
P(a,b | A, B, \rho_{AB})=\sum_\lambda p_\lambda P(a|A, \lambda)P_{Q}(b|B, \rho_{\lambda}),
\end{equation}
in which $P(a,b|A,B, \rho_{AB}) = Tr(A_{a} \otimes B_{b}\, \rho_{AB})$ is the probability
of getting outcomes $a$ and $b$ when measurements $A$ and $B $ are performed on $\rho_{AB}$ by Alice and Bob respectively, $A_{a}$ and $B_{b}$ are their corresponding measurement
operators;
$ \lambda $ is the hidden variable, $\rho_{\lambda}$ is the state that Alice sends with probability $p_{\lambda}$($ \sum_{\lambda}p_{\lambda} =1 $); $P(a|A, \lambda)$ is the conditioned probability of Alice obtaining outcome $a$ under   $\lambda$ , $P_{Q}(b|B, \rho_{\lambda})$ denotes the quantum probability of outcome $b$ given by measuring $B$ on the local hidden state $\rho_{\lambda}$.
Now, if Bob determines that such correlation $P(a,b | A, B, \rho_{AB})$ cannot be explained by any LHS models, then he will be convinced that Alice can steer his state, and thus the corresponding bipartite state is entangled. In short, the bipartite state $\rho_{AB}$ is unsteerable by Alice to Bob if and only if the joint probability distributions satisfy the Eq.(\ref{mon1}) for all measurements $A$ and $B$. The assumption of such LHS model leads to certain steering inequalities, violation of which indicates the occurrence of steering.\\
The simplest way of constructing steering inequality is to find constraint for the correlations between Alice's and Bob's measurement statistics. In this work, we are interested in using such type of linear steering inequality formulated by
Cavalcanti, Jones, Wiseman, and Reid(CJWR) \cite{Cav}. They proposed the following  series of steering inequalities to check whether a bipartite state is steerable from Alice to Bob when both the parties are allowed to perform $n$ dichotomic measurements on their respective subsystems:

\begin{equation}\label{mon2}
F_{n}(\rho_{AB},\mu) = \frac{1}{\sqrt{n}}|\sum_{k=1}^{n} \langle A_{k} \otimes B_{k} \rangle | \leq 1
\end{equation}
where $A_{k} = \hat{a}_{k}\cdot \overrightarrow{\sigma}$, $B_{k}  =  \hat{b}_{k} \cdot \overrightarrow{\sigma}$,  $\overrightarrow{\sigma} = (\sigma_{1},\sigma_{2},\sigma_{3})$ is a vector composed of the Pauli matrices, $\hat{a}_{k} \in \mathbb{R}^{3}$ are unit vectors,  $\hat{b}_{k} \in \mathbb{R}^{3}$  are orthonormal vectors, $\mu =\{\hat{a}_{1},\hat{a}_{2},....\hat{a}_{n}, \hat{b}_{1},\hat{b}_{2},...,\hat{b}_{n} \}$ is the set of measurement directions,  $\langle A_{k} \otimes B_{k} \rangle = Tr(\rho_{AB} (A_{k} \otimes B_{k}))$ and $\rho_{AB} \in \mathbb{H_{A}} \otimes \mathbb{H_{B}}$  is any bipartite quantum state.\\
Here our attention is confined to the qubit case. In Hilbert-Schmidt representation any two qubit state can be expressed as,

\begin{equation}\label{mon3}
\rho_{AB} = \frac{1}{4}[\mathbf{ I} \otimes \mathbf{ I}  + \vec{a } \cdot \vec{\sigma} \otimes \mathbf{ I} + \mathbb{I} \otimes \vec{b} \cdot \vec{\sigma} + \sum_{i,j} t^{AB}_{ij}  \sigma_{i}  \otimes \sigma_{j}]
\end{equation}
$\vec{a}$, $\vec{b}$ being the local bloch vectors and $T_{AB} = [t^{AB}_{ij}]$ is the correlation matrix. The components $t^{AB}_{ij}$ are given by $t^{AB}_{ij} = Tr[\rho_{AB} {\sigma}_{i}  \otimes {\sigma}_{j}]$ and $\vec{a } ^{2}+ \vec{b }^{2} + \sum_{i,j} {(t^{AB}_{ij})} ^{2}\leq 3$.
In \cite{Luo}, Luo showed that any two-qubit state
can be reduced, by local unitary equivalence, to

\begin{equation}\label{mon4}
\rho'_{AB} = \frac{1}{4}[\mathbf{ I} \otimes \mathbf{ I}  + \vec{a' } \cdot \vec{\sigma} \otimes \mathbf{ I} + \mathbb{I} \otimes \vec{b'} \cdot \vec{\sigma} + \sum_{i} u'_{i} \, {\sigma}_{i}  \otimes {\sigma}_{i}]
\end{equation}
where the correlation matrix of $\rho'_{AB} $ is $T'_{AB} = diag(u'_{1},u'_{2},u'_{3})$.
In \cite{Cos}, for any two qubit state $\rho'_{AB}$,  the authors derived an analytical expression for the maximum value of the two settings and three settings CJWR linear steering inequality in terms of diagonal elements of the correlation matrix $T'_{AB} = diag(u'_{1},u'_{2},u'_{3})$. \\

Specifically, $\max_{\mu} F_{2}(\rho'_{AB},\mu)$ and $\max_{\mu} F_{3}(\rho'_{AB},\mu)$ have been evaluated to be the following :
\begin{equation}\label{mon5}
\max_{\mu} F_{2}(\rho'_{AB},\mu) = \sqrt{u'^{2}_{1} + u'^{2}_{2} },
\end{equation}
and
\begin{equation}\label{mon6}
\max_{\mu} F_{3}(\rho'_{AB},\mu) = \sqrt{Tr[T'^{T}_{AB} T'_{AB}] },
\end{equation}
where ${u'}_{1}^{2}$ and ${u'}_{2}^{2}$ are two largest diagonal elements of $T'^{2}_{AB}$.
Here we do consider only the three settings linear steering inequality as under two measurement settings the notion of steering and Bell-CHSH nonlocality are indistinguishable.
Since the states given in Eq.(\ref{mon3}) and Eq.(\ref{mon4}) are local unitary equivalent, we must have,
\begin{eqnarray*}
\max_{\mu} F_{3}(\rho'_{AB},\mu) = \sqrt{Tr[T'^{T}_{AB} T'_{AB}] } &&=\sqrt{Tr[T^{T}_{AB} T_{AB}] } \\
                                                                                           &&=\max_{\mu} F_{3}(\rho_{AB},\mu).
\end{eqnarray*}

Consequently, the linear inequality(\ref{mon2}) (for three measurement settings) implies that any state $\rho_{AB}$ is $F_{3}$ steerable if and only if
\begin{equation}\label{mon8}
S_{AB} = Tr[T^{T}_{AB} T_{AB}] >1.
\end{equation}
Note that this condition is just a sufficient criterion to check steerability. There exist steerable states which satisfy $S_{AB} \leq 1$.

\section{Shareability and monogamy of Steering Correlations}\label{monsec1}
Consider a scenario in which Alice, Bob, and Charlie share a three qubit state $\rho_{ABC}.$ Let $\rho_{AB}$, $\rho_{AC}$, $\rho_{BC}$ denote the three bipartite reduced states of $\rho_{ABC}$.
In general, for tripartite states, monogamy relations have the following form:
\begin{equation}\label{mon9}
Q(\rho_{AB}) + Q(\rho_{AC}) \leq Q(\rho_{A|BC})
\end{equation}
or
\begin{equation}\label{mon10}
Q(\rho_{AB}) + Q(\rho_{AC}) \leq K
\end{equation}

for some bipartite quantum measure $Q$ and positive real number $K$. Here $Q(\rho_{A|BC})$ represents the correlation between subsystems $A$ and $BC$.
Entanglement, Bell-CHSH nonlocality, and steering (via two settings linear steering inequality) are examples of such correlation measures satisfying this monogamy relation(Eq.(\ref{mon10})).
Particularly, for Bell-CHSH inequality and $F_{2}$ inequality, monogamy relation(\ref{mon10})  takes the form \cite{Ton,Tne,See,Che}
\begin{equation}\label{mon11}
Q(\rho_{AB}) + Q(\rho_{AC}) \leq 2.
\end{equation}
Thus, at most one bipartite reduced state with respect to a certain observer (say, $A$) can violate the linear steering $F_{2}$ inequality. This feature is known as ``\textit{monogamy of steering correlations}''.\\
It is a known fact that entanglement is a property of a quantum state, now correlations generated due to measurements performed on any entangled quantum state are not solely determined by the state of the system under consideration. It is also dependent on the specific setup used to determine the correlations. Consequently, in general, steerability of a quantum state  varies from one measurement scenario to another.  In this context, an obvious question arises: \textit{can addition of one more observable per party change the monogamous nature of steering?}
Affirmative answer of this query is given by the following theorem.
\begin{theorem}\label{mont1}
For any three qubit state $\rho_{ABC} \in \mathbb{H}^{A} \otimes \mathbb{H}^{B} \otimes \mathbb{H}^{C}$, at most two reduced states can violate the three settings CJWR linear steering inequality, i.e steering can be non-monogamous when each party measures three dichotomic observables.
\end{theorem}
\begin{proof}
Any three qubit state $\rho_{ABC}$ can be represented as

\begin{align}\label{mon12}
&\rho_{ABC}  = \frac{1}{8}[\mathbb{I} \otimes \mathbb{I} \otimes \mathbb{I} + \vec{a} \cdot \vec{\sigma} \otimes \mathbb{I} \otimes \mathbb{I} + \mathbb{I} \otimes \vec{b} \cdot \vec{\sigma} \otimes \mathbb{I}\nonumber\\
& +\mathbb{I} \otimes \mathbb{I}  \otimes \vec{c} \cdot \vec{\sigma} + \sum_{ij} t^{AB}_{ij}  {\sigma}_{i}  \otimes {\sigma}_{j} \otimes \mathbb{I} +  \sum_{ik} t^{AC}_{ik}  {\sigma}_{i}  \otimes \mathbb{I}  \otimes {\sigma}_{k} \nonumber \\
& + \sum_{jk} t^{BC}_{jk}  \mathbb{I}  \otimes {\sigma}_{j}  \otimes {\sigma}_{k} +  \sum_{ijk} t^{ABC}_{ijk} {\sigma}_{i}  \otimes {\sigma}_{j}  \otimes {\sigma}_{k} ]
\end{align}
In the following we denote $\rho_{i}$ as the reduced density matrices for the subsystem $i = A,B,C$. One computes from Eq.(\ref{mon12}), that
\begin{equation}\label{mon13}
tr(\rho_{A}^{2}) = \frac{1+ \vec{a}^{2}}{2}, Tr(\rho_{BC}^{2}) = \frac{1}{4}(1+ \vec{b}^{2} + \vec{c}^{2} + S_{BC}).
\end{equation}
Similarly,
\begin{eqnarray}\label{mon14}
&&tr(\rho_{B}^{2}) = \frac{1+ \vec{b}^{2}}{2}, Tr(\rho_{AC}^{2}) = \frac{1}{4}(1+ \vec{a}^{2} + \vec{c}^{2} +S_{AC}),\nonumber\\
&&tr(\rho_{C}^{2}) = \frac{1+ \vec{c}^{2}}{2}, Tr(\rho_{AB}^{2}) = \frac{1}{4}(1+ \vec{a}^{2} + \vec{b}^{2} + S_{AB}).
\end{eqnarray}
First consider $\rho_{ABC}$ is a pure state. Then from Schimdt decomposition, we have  $Tr(\rho_{i}^{2}) = Tr(\rho_{jk}^{2})$ for $i \neq j \neq k$, $i, j, k = A,B,C$.
Using these relations and Eqs.(\ref{mon13},\ref{mon14}), it is straightforward to calculate $S_{ij}$ of each pair of qubits, yielding:
\begin{equation}\label{mon15}
S_{AB} = 1 + 2 \vec{c}^{2} - \vec{a}^{2} - \vec{b}^{2},
\end{equation}

\begin{equation}\label{mon16}
S_{AC} = 1 + 2 \vec{b}^{2} - \vec{a}^{2} - \vec{c}^{2},
\end{equation}
and
\begin{equation}\label{mon17}
S_{BC} = 1 + 2 \vec{a}^{2} - \vec{b}^{2} - \vec{c}^{2}.
\end{equation}
Adding these three relations and simplifying it, one obtains the following relation:
\begin{equation}\label{mon18}
S_{AB} + S_{AC} + S_{BC} = 3.
\end{equation}
This relation is derived by the similar method used in \cite{Ken} for developing Bell monogamy relations.\\
Now, taking mixed state $\rho_{ABC}$ as $\sum_{n} p_{n}|\psi_{n}\rangle \langle \psi_{n}|$, one has
$S_{AB} \leq \sum_{n} p_{n} S^{n}_{AB}$, and similar relations for $S_{AC}$, $S_{BC}$.
Adding these relations and using Eq.(\ref{mon18}), we obtain
\begin{equation}\label{mon19}
S_{AB} + S_{AC} + S_{BC} \leq 3
\end{equation}
This is a trade off relation among two qubits of any three qubit state $\rho_{ABC}$. Now $S_{AB} >1$ is sufficient for Alice and bob to witness violation of $F_{3}$ inequality. Hence, inequality Eq.(\ref{mon19}) imposes constraint on quantum steering: it is impossible that all pair of qubits violate $F_{3}$ inequality. \\
But the trade off relation (\ref{mon19}) is unable to assure us about the number of two qubit reduced states that can violate $F_{3}$ inequality.
To complete the proof, we still have to find two reduced states of $\rho_{ABC}$ which violate $F_{3}$ inequality. \\
Using Eqs.(\ref{mon15},\ref{mon16},\ref{mon17}), one can easily find that the reduced states $\rho_{AB}$  and $\rho_{AC}$ of the pure three qubit state $\rho_{ABC}$ will violate $F_{3}$ inequality iff the following inequality is satisfied :
\begin{equation}\label{mon20}
\vec{c}^{2} > \frac{\vec{a}^{2} + \vec{b}^{2}}{2},
\vec{b}^{2} > \frac{\vec{a}^{2} + \vec{c}^{2}}{2},
\end{equation}
One can similarly obtain condition of violation for other pairs of reduced states.
Now, consider the fully entangled three qubit state,
\begin{equation}\label{mon21}
|\psi_{ABC}\rangle = \frac{1}{2}(|100\rangle +|010\rangle + \sqrt{2} |001\rangle).
\end{equation}
By using the above conditions, one can find that bipartite correlations between party A and C of subsystem AC and between B and C of subsystem BC violate the $F_{3}$ inequality: $S_{BC} =  S_{AC}  = 1+ \frac{1}{4}>1$. This shows that some of the steering correlations between party A and C can thus be shared with party B and C. Thus, under some conditions (for example, Eq. (\ref{mon20}) and its permutations), steering is non-monogamous with respect to $F_{3}$ inequality.
\end{proof}
The above result on symmetric states leads to the following corollary.
\begin{corollary}\label{monc1}
None of the three reduced states of any three qubit symmetric state $\rho_{ABC}$ violates $F_{3}$ inequality i.e steering is monogamous for such states with respect to $F_{3}$ inequality.
\end{corollary}
\begin{figure}
\includegraphics[scale=0.32]{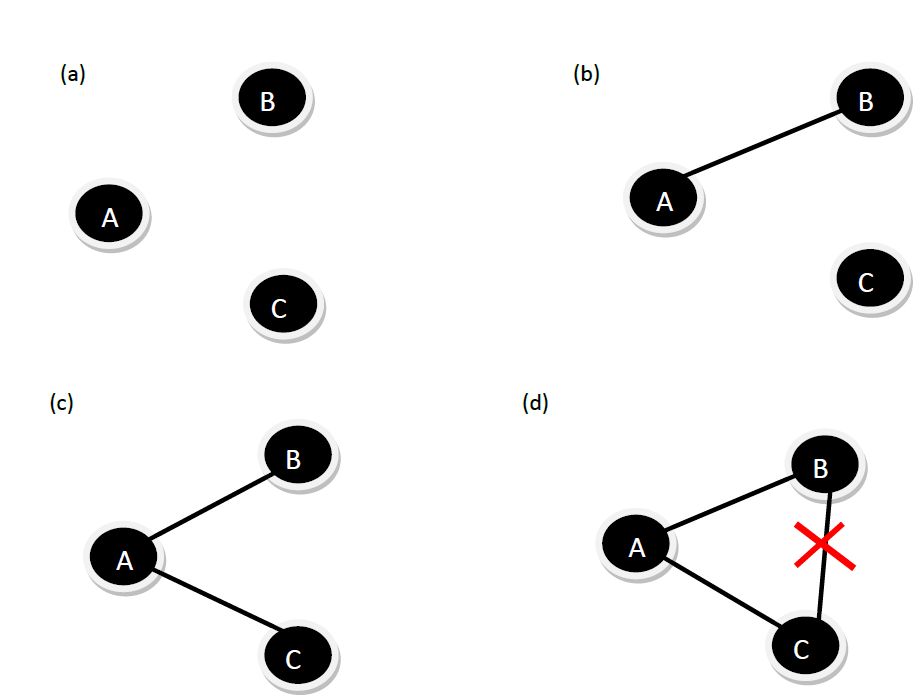}
\caption{Steering Graphs: Here each circle represents physical system and a solid line connecting two systems describes the bipartite steering correlation between them. Different possibilities of sharing bipartite steering among three distant physical systems are depicted in this figure. (a) No bipartite steering is detected between individual parties. For example, the tripartite GHZ state  \cite{Ghz} $|\phi_{GHZ}\rangle = \frac{|000\rangle + |111\rangle}{\sqrt{2}}$ has no bipartite steering. (b) bipartite steering of one reduced state is detected. One such kind of state is pure biseparable state. (c) Two bipartite reduced states are steerable. As we have shown in sec.(\ref{monsec1}), the state $|\psi_{ABC}\rangle$ belongs to this group. (d) Trade-off relation Eq.(\ref{mon19}) prevents bipartite steering between every pair of systems.}\label{monp}
\end{figure}
 Theorem\ref{mont1} guarantees existence of three qubit states for which all two-party
reduced states with respect to a certain observer violate the $F_{3}$ inequality(Fig.\ref{monp}). This non-monogamous nature of steering allows one to reveal shareability(non-monogamous nature) of the entanglement of bipartite mixed states. As far as the shareability of quantum correlations is concerned, quantum entanglement is strictly speaking only monogamous in the case of pure entangled states. If the state of two systems, say $\rho_{AB}$ is a mixed entangled state, then it is possible that both of the two systems $A$ and $B$ are entangled to a third system, say $C$. For example, the so-called $W$ state \cite{Dur} $|W\rangle = \frac{(|001\rangle + |010\rangle + |100\rangle)}{\sqrt{3}}$ has bipartite reduced states that are all identical and entangled. Thus, entanglement of these reduced bipartite mixed states is sharable (non-monogamous), however, the steering correlations obtainable from these states follow the monogamy inequality Eq.(\ref{mon11}). So, by considering $F_{2}$ inequality, one cannot reveal shareability of entanglement of bipartite mixed states. To reveal this, steering correlations obtainable from these states must be non-monogamous. As shown above in Theorem\ref{mont1}, the state $|\psi_{ABC}\rangle$ (Eq.(\ref{mon21})) provides steerable bipartite reduced states between subsystems $AC$ and $BC$. Therefore the corresponding reduced mixed states $\rho_{AC}$ and $\rho_{BC}$ are also entangled and the two qubit mixed entangled state $\rho_{AC}$ is shareable to at least one other qubit. This in turn indicates that $F_{3}$ inequality is an appropriate ingredient to reveal shareability of entanglement of mixed states.\\
Unlike the standard $|W\rangle$ state, the state $|\psi_{ABC}\rangle$ can be used as a resource for deterministic teleportation and dense coding \cite{Pat}. As another application of the non-monogamous nature of steering correlations, consider that a pure three qubit state is provided to experimentalists which they have to use as a resource in deterministic teleportation or dense coding. They are also provided with the information that the state is either $|\psi_{ABC}\rangle$ or $|W\rangle.$ We show that the non-monogamy phenomenon as described in theorem\ref{mont1} can be used to determine the desired state. For $|W\rangle$ state, $S_{ij} = 1$ for all reduced states, so monogamy is preserved. On the other hand, the state $|\psi_{ABC}\rangle$ does not follow the monogamy as shown in theorem \ref{mont1}. Thus, the above result distinguishes the two types of states though they belong to the same class ($W$-like states \cite{Dur}). \\
 Keeping in mind the usefulness of shareability relations, one naturally would be interested to know which of the three qubit states obey monogamy(or non-monogamy) of steering. The explicit evaluation of the number of reduced steerable states along with monogamy(or non-monogamy) in each class of three qubit pure states as classified by Sab\'{i}n and Garc\'{i}a-Alcaine\cite{Sab} is reported in Appendix \ref{Classi}, where we see that only star shaped states and $W$-like states can be non-monogamous. Next we ask whether non-monogamous behavior of those two classes of pure states is robust against white noise admixture. The results are presented in Appendix \ref{Classi}, where it is shown that less entangled states can be more robust against white noise admixture in comparison to higher entangled states. \\
 Other than constraint given by Eq. (\ref{mon20}) and its permutations, few other conditions are also derived in the following sections under which $F_{3}$ steering is non-monogamous.
\section{Reduced Steering Versus Entanglement}\label{monsec3}
In two qubit systems, the more entangled a pure state is, the more it can violate the Bell-CHSH inequality. In this context, a relevant study is to find the relation between violation of $F_{3}$ inequality by the reduced bipartite states of a pure state and their corresponding entanglement(with respect to some measure).
The relation between $S_{AB}$ and concurrence $\mathcal{C}_{AB}$(a measure of entanglement)\cite{Hil,Wot} can be derived with similar methods used in \cite{Ves}. For pure bipartite states the relation is $S_{AB} = 1+2 \, \mathcal{C}_{AB}^{2}$. Hence, for pure states, more entanglement generates  larger violation of $F_{3}$ inequality. However from this relation we cannot infer anything about mixed bipartite reduced states of a three qubit pure state. In the theorem below we derive a relation justifying our claim.

\begin{theorem}\label{mont4}
The triples $(S_{AB},S_{AC},S_{BC})$ of three reduced states obtained from a pure three qubit state and $(\mathcal{C}_{AB},\mathcal{C}_{AC},\mathcal{C}_{BC})$ maintain the same ordering i.e.,
\begin{equation}\label{mon24}
S_{AB}> S_{AC} > S_{BC}\,\,\,\,\, \text{iff} \,\,\,\, \mathcal{C}_{AB} > \mathcal{C}_{AC} > \mathcal{C}_{BC}.
\end{equation}
\end{theorem}
\begin{proof}
By eliminating $\vec{a}$ from Eq.(\ref{mon15}) and Eq.(\ref{mon16}), we have
\begin{equation}\label{mon25}
S_{AB} - S_{AC} = 3 (\vec{c}^{2} - \vec{b}^{2}).
\end{equation}
Now, three tangle $\tau$ \cite{Cof}, for three qubit pure state, is given by\cite{Che}
\begin{align}\label{mon26}
&&\tau= 1 - \vec{a}^{2} - \mathcal{C}_{AB}^{2} - \mathcal{C}_{AC}^{2}\nonumber \\
       &&= 1 - \vec{b}^{2} - \mathcal{C}_{AB}^{2} - \mathcal{C}_{BC}^{2}\nonumber \\
       &&= 1 - \vec{c}^{2} - \mathcal{C}_{AC}^{2} - \mathcal{C}_{BC}^{2}
\end{align}
Comparing these equalities, we obtain
\begin{equation}\label{mon27}
\mathcal{C}_{AB}^{2} - \mathcal{C}_{AC}^{2} = \vec{c}^{2} - \vec{b}^{2}
\end{equation}
and its permutations, which immediately lead to
\begin{equation}\label{mon28}
S_{AB} - S_{AC} = 3 (\mathcal{C}_{AB}^{2} - \mathcal{C}_{AC}^{2})
\end{equation}
and its permutations. Thus, we have developed the ordering relation as per Eq.(\ref{mon24}).
\end{proof}
It is interesting to note that $(S_{AB},S_{AC},S_{BC})$ and $(\vec{c}^{2}, \vec{b}^{2}, \vec{a}^{2} )$ follow the same ordering for all pure three qubit state.\\
Distribution of bipartite quantum entanglement (i.e., entanglement of reduced bipartite states) of any pure three qubit state is subjected to certain shareability laws. In particular, addition of squared concurrence of all bipartite reduced states cannot be greater than $\frac{4}{3}$\cite{Dur},
\begin{equation}\label{mon100}
\mathcal{C}_{AB}^{2} + \mathcal{C}_{AC}^{2} +  \mathcal{C}_{BC}^{2} \leq \frac{4}{3}.
\end{equation}
This shareability constraint indicates that shareability of reduced bipartite steerability as well as individual bipartite steerability of any pure three qubit state might depend on concurrence of reduced bipartite states. This is in fact the case. We next discuss few results in this direction.

\begin{theorem}\label{mont4}
If the squared concurrence of any bipartite reduced state for a pure three qubit state is greater than  $\frac{4}{9}$, then the corresponding reduced state is $F_{3}$ steerable i.e.,
if $\mathcal{C}_{ij}^{2} > \frac{4}{9}$ $(i \neq j, i, j = A, B, C)$, then the corresponding reduced state $\rho_{ij}$ is $F_{3}$ steerable.
\end{theorem}
\begin{proof}
By using  Eqs. (\ref{mon27}) and (\ref{mon15},\ref{mon16},\ref{mon17}), each of $S_{ij}$ can be expressed in terms of $\mathcal{C}_{ij}$,
 \begin{equation}\label{mon99}
S_{AB} = 1 + 2 \, \mathcal{C}_{AB}^{2} - \mathcal{C}_{AC}^{2} - \mathcal{C}_{BC}^{2}
\end{equation}
and its permutations.
Let, $\mathcal{C}_{AB}^{2} = \frac{4}{9} + \epsilon$, where $\epsilon$ is sufficiently small positive number. This immediately restricts the sum of squared concurrence of other two bipartite reduced states,

\begin{equation}\label{mon98}
\mathcal{C}_{AC}^{2} + \mathcal{C}_{BC}^{2} \leq \frac{8}{9} - \epsilon.
\end{equation}
Applying these to the expression of $S_{AB}$, this leads to the sharp inequality $S_{AB} \geq 1 + \epsilon$. So, if $\mathcal{C}_{AB}^{2} > \frac{4}{9}$, the $F_{3}$ inequality is violated. Similarly, it can be proved for other bipartite reduced states.
\end{proof}
This result holds for all three qubit pure states. As an example, consider the pure state $|\psi_{ABC}\rangle$ which has two $F_{3}$ steerable reduced states $\rho_{AC}$ and $\rho_{BC}$ with $\mathcal{C}_{AC}^{2}  = \mathcal{C}_{BC}^{2} = \frac{1}{2} > \frac{4}{9}$.
However, one should note that the above inequality $\mathcal{C}_{ij}^{2} > \frac{4}{9}$ is only a sufficient condition for $F_{3}$ steerability of reduced bipartite state $\rho_{ij}$, because there are reduced states which violate  the inequality $\mathcal{C}_{ij}^{2} > \frac{4}{9}$, but still give rise to $F_{3}$ steerability. One such example is $|\phi_{con}\rangle = \frac{\sqrt{3}}{2} |000\rangle + \frac{1}{2\sqrt{2}} |101\rangle + \frac{1}{2\sqrt{2}} |110\rangle$. For this state, from the above formulae one can obtain $\mathcal{C}_{AB}^{2} = \frac{3}{8} < \frac{4}{9}$ and $S_{AB} = 1 + \frac{5}{16}$. Clearly, the reduced state $\rho_{AB}$ violates $F_{3}$ inequality, while it violates the inequality $\mathcal{C}_{AB}^{2} > \frac{4}{9}$.
Although, an obvious necessary and sufficient condition can be derived from Eq.(\ref{mon99}) and its permutations.
\begin{corollary}\label{monc2}
Any reduced state $\rho_{ij}$ of a three qubit pure state will violate $F_{3}$ inequality if and only if squared concurrence of the corresponding reduced state is greater than the average of the squared concurrence of the remaining two reduced states,i.e.,\\
$S_{ij}>1$ if and only if $\mathcal{C}_{ij}^{2} > \frac{\mathcal{C}_{ik}^{2} + \mathcal{C}_{jk}^{2}}{2}$, where $ i \neq j \neq k$ and $i, j, k = A, B, C$.
\end{corollary}
Due to the shareability constraint Eq.(\ref{mon100}), violation of one of the reduced states (say, $\rho_{AB}$) puts strong restriction on the
average of squared concurrences of the remaining reduced states.
\begin{corollary}\label{monc3}
For any $F_{3}$ steerable reduced state $\rho_{ij}$, the following inequality holds : \\
$\frac{\mathcal{C}_{ik}^{2} + \mathcal{C}_{jk}^{2}}{2} < \frac{4}{9}$, where $ i \neq j \neq k$ and $i, j, k = A, B, C$.
\end{corollary}
As shown in \cite{Cof}, sum of squared concurrence between $i$ and $k$, and the squared concurrence between $j$ and $k$, cannot be greater than $1$, i.e., $\mathcal{C}_{ik}^{2} + \mathcal{C}_{jk}^{2} \leq 1$. Hence, from the above corollary it is observed that this restriction is further strengthened if one consider $F_{3}$ steerability of $\rho_{ij}$. \\
Since the last corollary (\ref{monc3})  puts more stringent restriction, using it we get the following sufficient condition for monogamy of $F_{3}$ steering:
\begin{corollary}\label{monc4}
For any pure three qubit state $\rho_{ABC}$, steering correlations will obey monogamy if  $\mathcal{C}_{ik}^{2} + \mathcal{C}_{jk}^{2} \geq \frac{8}{9}$, where $ i \neq j \neq k$ and $i, j, k = A, B, C$ holds for at least two of three possible cases.
\end{corollary}
It may be noted that theorem \ref{mont4}, gives rise to a sufficient condition for non-monogamy of $F_{3}$ steerability.
\begin{corollary}\label{monc5}
$F_{3}$ steering is non-monogamous if $\mathcal{C}_{ij}^{2} > \frac{4}{9}$ ($ i \neq j$ and $i, j = A, B, C$) for any two pairs of $i$, $j$.
\end{corollary}

Now, we discuss how the $F_{3}$ inequality violation by the reduced bipartite states depends on the genuine entanglement of the three qubit state. As shown in Sec.(\ref{monsec1}), maximum two bipartite reduced states of $\rho_{ABC}$ can violate the $F_{3}$ inequality, so the bipartite steering of $\rho_{ABC} $ implies that it comes from one component of either this triple $(S_{AB}, S_{AC}, S_{BC})$ or $((S_{AB}, S_{AC}),(S_{AB}, S_{BC}),(S_{AC},S_{BC}))$. Considering both the possibilities, we adopt two different measures: $S^{\max}(\rho_{ABC})$  and $S_{total}^{max}(\rho_{ABC})$, where
$S^{\max}(\rho_{ABC}) = \max{\{S_{AB},S_{AC},S_{BC}\}}$ and $S_{total}^{\max}(\rho_{ABC}) = \max{\{S_{AB}+ S_{AC}, S_{AB} + S_{BC}, S_{AC}+ S_{BC}\}}$.\\
In each case, we will now derive a relation with tripartite entanglement of $\rho_{ABC}$.
\begin{theorem}\label{mont5}
For an arbitrary three qubit state $\rho_{ABC}$, the three tangle $\tau(\rho_{ABC})$ and maximum bipartite steering($S^{\max}(\rho_{ABC})$) with respect to $F_{3}$ inequality obeys the following complementary relation :
\begin{equation}\label{mon29}
S^{max}(\rho_{ABC}) + 2 \tau(\rho_{ABC}) \leq 3.
\end{equation}
\end{theorem}
\begin{proof}
Note that for pure three qubit state Eq.(\ref{mon15}) provides $S_{AB} = 1 + 2 \vec{c}^{2} - \vec{a}^{2} - \vec{b}^{2}$.
Incorporating this with the third equality of three tangle in Eq.(\ref{mon26}), we obtain
\begin{align}\label{mon30}
 S_{AB} + 2 \tau(\rho_{ABC}) &= 3 - \vec{a}^{2} - \vec{b}^{2} - 2 \mathcal{C}_{AC}^{2} - 2 \mathcal{C}_{BC}^{2} \nonumber\\
& \leq 3
\end{align}
Similarly, one has $S_{AC} + 2 \tau(\rho_{ABC}) \leq 3 $ and $S_{BC} + 2 \tau(\rho_{ABC}) \leq 3$.
Hence for pure state $S^{max}(\rho_{ABC}) + 2 \tau(\rho_{ABC}) \leq 3.$
As the three tangle $\tau$ and $S^{max}(\rho_{ABC})$ both are convex under mixing, it implies that the relation in Eq.(\ref{mon29}) holds for all three qubit states.
\end{proof}
This complementary relation suggests that $F_{3}$ inequality violation by the reduced bipartite states depends on the tripartite entanglement present in the tripartite system.
We determine a class of three qubit genuinely entangled states which saturates the above-mentioned relation. This single parameter class of states is given by $|\phi_{m}\rangle = \frac{|000\rangle +m(|101\rangle + |010\rangle) + |111\rangle}{\sqrt{2 + 2m^{2}}}$, where $m \in [0,1]$. The above class of states has been identified in \cite{Nep} as the maximum dense-coding capable class of states. For this class of states, $S^{max}(|\phi_{m}\rangle)= 1 + \frac{8 m^{2}}{(1+m^{2})^{2}}$ and $\tau(|\phi_{m}\rangle) = 1 - \frac{4 m^{2}}{(1+m^{2})^{2}}$.  Hence, for this class of states, one can show the follwing relation :  $S^{max}(|\phi_{m}\rangle) + 2 \tau(|\phi_{m}\rangle) = 3$.

\begin{theorem}\label{mont6}
For an arbitrary three qubit state $\rho_{ABC}$, the three tangle $\tau(\rho_{ABC})$ and maximum bipartite steering($S^{\max}_{total}(\rho_{ABC})$) satisfy the following complementary relation:
\begin{equation}\label{mon80}
S^{max}_{total}(\rho_{ABC}) + \tau(\rho_{ABC}) \leq 3.
\end{equation}
\end{theorem}
\begin{proof}
Combining  Eqs.(\ref{mon15},\ref{mon16})  and last two equalities of Eq.(\ref{mon26}), we get:
\begin{align*}
 S_{AB} + S_{AC} + 2 \tau(\rho_{ABC}) & = 4 - 2 \vec{a}^{2} - \mathcal{C}_{AC}^{2} -  \mathcal{C}_{AC}^{2} - 2\, \mathcal{C}_{BC}^{2} \nonumber\\
&= 3 + \tau - \vec{a}^{2} - 2\, \mathcal{C}_{BC}^{2}.\nonumber\\
\end{align*}
Thus,
\begin{equation}\label{mon31}
S_{AB} + S_{AC} + \tau(\rho{ABC}) \leq 3.
\end{equation}
Considering all permutation of parties, we get $S_{AB} + S_{BC} + \tau(\rho{ABC}) \leq 3 $ and $S_{AC}  + S_{BC} +  \tau(\rho_{ABC}) \leq 3$.\\
Now, by using the convexity property of the left hand sides of these inequalities, we claim that the relation(Eq.(\ref{mon29})) holds for all three-qubit states.
\end{proof}
We have identified a class of genuinely entangled states which saturates the afore-mentioned relation. This class of states is given by $|\phi_{q}\rangle = \frac{1}{\sqrt{2}}|000\rangle + \sqrt{\frac{1}{2}-q^{2}} |101\rangle + q |111\rangle$ where $q \in  (0, \frac{1}{\sqrt{2}})$.
For $|\phi_{q}\rangle$, $S_{total}^{max} = 3 - 2q^{2}$ and  $\tau = 2 q^{2}$. Hence, $S^{max}_{total}(\rho_{ABC}) + \tau(\rho_{ABC})= 3$. However, $|\phi_{q}\rangle$ has only one reduced state which violate $F_{3}$ inequality. Since, among all pure three qubit GHZ class of  states only star shaped states can have two reduced steerable states (see Appendix \ref{Classi}) and for this class of states, $S^{max}_{total}(\rho_{ABC}) + \tau(\rho_{ABC}) < 3$, so there is no three qubit pure state with $\tau \neq 0$ having two reduced bipartite steerable states which saturates the above inequality.\\
All the above-mentioned relations are obtained with respect to three tangle. However, three tangle is not a good measure of genuine tripartite entanglement even for pure states as there exists a large number of pure states ($W$-like states\cite{Dur}) for which it becomes zero. Hence, none of the relations are meaningful for those $W$-like states.\\
To obtain such relations for $W$-like states, we consider the measure for  $W$ entanglement introduced by Dur et. al.  \cite{Dur}, defined as
$E_{W} = \min\{\mathcal{C}_{AB}^{2},\mathcal{C}_{AC}^{2},\mathcal{C}_{BC}^{2}\}$.
Any pure state $\rho_{ABC}$ contains $W$ entanglement if $E_{W} > 0$. The $W$ entanglement $E_{W}$ achieves its maximum value $\frac{4}{9}$ in the $|W\rangle$ state.
 \begin{theorem}\label{mont7}
For an arbitrary three qubit  pure state $|\phi_{ABC}\rangle$, the $W$ entanglement $(E_{W})$ and maximum bipartite steering($S^{\max}_{total}(\rho_{ABC})$) satisfies the following complementary relation:
\begin{equation}\label{mon32}
S^{max}_{total}(|\phi_{ABC}\rangle) + 3 E_{W}(|\phi_{ABC}\rangle) \leq \frac{10}{3}.
\end{equation}
\end{theorem}
\begin{proof}
Using Eq.(\ref{mon99}) and its permutations, we have
\begin{equation}\label{mon33}
S_{AB} + S_{AC} + 3\, \mathcal{C}_{BC}^{2} = 2 + \mathcal{C}_{AB}^{2} + \mathcal{C}_{AC}^{2} + \mathcal{C}_{BC}^{2}.
\end{equation}
If one uses Eq.(\ref{mon100}), the above equality immediately leads to
\begin{equation}\label{mon34}
S_{AB} + S_{AC} + 3 \, \mathcal{C}_{BC}^{2} \leq \frac{10}{3}.
\end{equation}
Similarly, permutation of parties gives $S_{AB} + S_{BC} + 3\, \mathcal{C}_{AC}^{2} \leq \frac{10}{3}$, and $S_{AC} + S_{BC} + 3 \, \mathcal{C}_{AB}^{2} \leq \frac{10}{3}$.
The above equations confirm the validity of the claim made in Eq.(\ref{mon32}).
\end{proof}
This relation imposes a restriction on the bipartite steering for a given amount of $W$ entanglement and it is saturated by $|W\rangle$ state. \\
We have also investigated such complementary relations for bipartite nonlocality(with respect to Bell-CHSH violation), bipartite steering and three tangle. Following the same procedure as before, a similar trade-off relation can be obtained for them:

\begin{equation}\label{mon34}
S^{\max}(\rho_{ABC}) + \mathcal{M}^{\max}(\rho_{ABC}) + 3 \,\tau(\rho_{ABC}) \leq 5
\end{equation}
where $\mathcal{M}^{\max}(\rho_{ABC}) = \max \{\mathcal{M}_{AB}, \mathcal{M}_{AC}, \mathcal{M}_{BC}\}$ and $\mathcal{M} = u_{1}^{2} + u_{2}^{2}$ is the Horodecki parameter \cite{Hro} used for measuring the degree of Bell-CHSH violation. $u_{1}^{2}, u_{2}^{2}$ are being the largest two eigen values of $T^{T}_{AB} T_{AB}$.

\section{Complementary relations for local and nonlocal information contents}
Total information content of a three qubit state can be divided into two forms: local and nonlocal information contents. Local information can be defined as \cite{Zha} :
\begin{equation}\label{mon35}
I_{local} = \vec{a}^{2} + \vec{b}^{2} + \vec{c}^{2}.
\end{equation}
To derive the complementary relation between local and nonlocal information contents, we consider only bipartite nonlocal information present in the three qubit state. Bipartite nonlocal information content can be defined as,
\begin{equation}\label{mon36}
I_{nonlocal} = \max\{N_{AB} + N_{AC}, N_{AB} + N_{BC}, N_{AC} + N_{BC}\}
\end{equation}
where $N_{ij} = \max\{0, S_{ij} - 1\}$, $i \neq j$ and $i, j = A, B, C$, quantifies the amount of $F_{3}$ inequality violation and hence the steering nonlocal correlations of the two qubit state $\rho_{ij}$ .
\begin{theorem}
For an arbitrary three qubit state $\rho_{ABC}$,
\begin{equation}\label{mon37}
I_{local} + I_{nonlocal} \leq 3.
\end{equation}
\end{theorem}
\begin{proof}
For pure three qubit states, it is straightforward to check that
\begin{align}\label{mon38}
I_{local} + (S_{AB} -1) + (S_{AC} - 1) &= 2(\vec{b}^{2} + \vec{c}^{2})- \vec{a}^{2}\nonumber \\
       &\leq 2(1+\vec{a}^{2}) - \vec{a}^{2}\nonumber \\
       & \leq 3
\end{align}
where, in the first inequality we have used the fact that relation $\vec{b}^{2} + \vec{c}^{2} \leq 1 + \vec{a}^{2}$ holds for all pure three qubit states \cite{Hig}.
Since $I_{local} \leq 3$, the above inequality(Eq.(\ref{mon38})) also holds when both of $N_{AB}$ and $N_{AC}$ are equal to zero. Hence $I_{local} + N_{AB} + N_{AC} \leq 3$. Similarly, one gets $I_{local} + N_{AB} + N_{BC} \leq 3$ and $I_{local} + N_{AC} + N_{BC} \leq 3$.
Note that the left hand sides of these inequalities are convex under mixing. This confirms the relation presented in Eq.(\ref{mon37}).
\end{proof}

The above trade-off relation links between local information and bipartite steering. One can easily show that $I_{local} = 3$, and $I_{nonlocal} = 0$ for the product state. On the other hand, in order to exist bipartite steering, $I_{local}$ must be less than $3$. For $|\psi_{ABC}\rangle$(Eq.(\ref{mon21})), $I_{local} = \frac{1}{2}$,  $I_{nonlocal} = 2 + \frac{1}{2}$ and it is the state which saturates this trade-off.
In this context, it may be noted that to get larger violation of $F_{3}$ inequality(characterizing larger amount of steering), the amount of local information content must be reduced. This fact will be confirmed in the next section, where we will show that amount of local information content must be less than one for any three qubit pure state to have two $F_{3}$ steerable bipartite reduced states.
\section{Entanglement detection}\label{monsec2}
We now illustrate the relevance of the above results with
some applications. By using the shareability relations, we will derive criteria of detecting different types of tripartite entanglement.
\begin{theorem}\label{mont2}
For any three qubit pure state $|\phi_{ABC}\rangle$$ \in$$ \mathbb{H}^{A} \otimes \mathbb{H}^{B} \otimes \mathbb{H}^{C}$, if at least one of the following conditions holds:
\begin{equation}\label{mon22}
(i)\vec{a}^{2} \neq \frac{\vec{b}^{2} + \vec{c}^{2}}{2},
(ii)\vec{b}^{2} \neq \frac{\vec{a}^{2} + \vec{c}^{2}}{2},
(iii)\vec{c}^{2} \neq \frac{\vec{a}^{2} + \vec{b}^{2}}{2}
\end{equation}
then the state is entangled.
\end{theorem}
\begin{proof}
Let $|\phi_{ABC}\rangle$ be a separable state, then all bipartite reduced states are also separable and $S_{AB}, S_{AC}, S_{BC} \leq 1.$ Hence violation of $F_{3}$ inequality by any bipartite reduced state entails entanglement of $|\phi_{ABC}\rangle$.
It is clear from Eq.(\ref{mon15})-Eq.(\ref{mon17} that if  $S_{AB}$, $S_{AC}$ and $S_{BC}>1$ then
$\vec{c}^{2} > \frac{\vec{a}^{2} + \vec{b}^{2}}{2}$, $\vec{b}^{2} > \frac{\vec{a}^{2} + \vec{c}^{2}}{2}$ and $\vec{a}^{2} > \frac{\vec{b}^{2} + \vec{c}^{2}}{2}$ hold respectively.
Again by adding Eq.(\ref{mon15}) and Eq.(\ref{mon16}), we have $S_{AB} + S_{AC} = 2 + \vec{b}^{2} +  \vec{c}^{2} - 2 \vec{a}^{2} $. By noting that $S_{AB} + S_{AC}> 2$ implies steerability of at least  one of $\rho_{AB}$ or $\rho_{AC}$,  $|\phi_{ABC} \rangle$ is entangled if $\vec{a}^{2} < \frac{\vec{b}^{2} + \vec{c}^{2}}{2}$. Similarly permutation of the parties gives $\vec{b}^{2} < \frac{\vec{a}^{2} + \vec{c}^{2}}{2}$ and $\vec{c}^{2} < \frac{\vec{a}^{2} + \vec{b}^{2}}{2}$. Combining all these expressions, we arrive at Eq.(\ref{mon22})
\end{proof}
Now one may enquire whether condition (\ref{mon22}) is also necessary for entanglement. Unfortunately, this is not the case. For example, consider the $|W\rangle$ state, which does not satisfy (\ref{mon22}), but is entangled.

\begin{theorem}\label{mont3}
For any three qubit pure state $|\phi_{ABC}\rangle $$\in$$ \mathbb{H}^{A} \otimes \mathbb{H}^{B} \otimes \mathbb{H}^{C}$, if at least one of the following conditions holds:
\begin{align}\label{mon23}
(i) && \vec{a}^{2} > \frac{\vec{b}^{2} + \vec{c}^{2}}{2},
     \vec{b}^{2} > \frac{\vec{a}^{2} + \vec{c}^{2}}{2},\nonumber \\
(ii) && \vec{a}^{2} > \frac{\vec{b}^{2} + \vec{c}^{2}}{2},
      \vec{c}^{2} > \frac{\vec{a}^{2} + \vec{b}^{2}}{2},\nonumber \\
(iii) && \vec{b}^{2} > \frac{\vec{a}^{2} + \vec{c}^{2}}{2},
      \vec{c}^{2} > \frac{\vec{a}^{2} + \vec{b}^{2}}{2} \nonumber \\
\end{align}
then the state is genuinely entangled.
\end{theorem}
\begin{proof}
Let $|\phi_{ABC}\rangle$ be any bi-separable state in which $AB$ is independent of $C$, then  it can be expressed as $(\cos\theta |00\rangle + \sin\theta |11\rangle)_{AB} \otimes |0\rangle_{C}$ where $0\leq \theta \leq \frac{\pi}{4}$. For this state, $\vec{c}^{2} = 1$ and $\vec{a}^{2} = \vec{b}^{2}$.  Using Eq.(\ref{mon15})-Eq.(\ref{mon17}), one can find that $S_{AB} = 3 -  2 \vec{a}^{2},  S_{AC} = \vec{a}^{2}, S_{BC} = \vec{a}^{2}$. So, only $S_{AB}$ can be greater than $1$. Similarly, one can show that only one reduced state will violate the $F_{3}$ inequality in which another system other than C system factorizes. This immediately leads to a simple sufficient condition for genuinely entangled pure states: Violation of $F_{3}$ inequality by two reduced states indicates genuine entanglement of $|\phi_{ABC}\rangle$. Then, from
Eq.(\ref{mon15})-Eq.(\ref{mon17}), we obtain the conditions (\ref{mon23}).
\end{proof}
It is important to note that for a pure biseparable state $\vec{a}^{2} + \vec{b}^{2} + \vec{c}^{2} \geq 1$ and exactly one of the reduced bipartite states is $F_{3}$ steerable. Therefore, for the existence of two $F_{3}$ steerable bipartite reduced states of a three qubit pure state, $\vec{a}^{2} + \vec{b}^{2} + \vec{c}^{2} < 1$ must hold. This condition can be treated as a necessary condition for a three qubit pure state to have two $F_{3}$ steerable bipartite reduced states. However, this is not sufficient, for example, $\vec{a}^{2} + \vec{b}^{2} + \vec{c}^{2} = \frac{1}{3} < 1$ for $|W\rangle$ state, but no reduced bipartite state of this state is $F_{3}$ steerable.\\
At this stage a pertinent question would be whether there exists any biseparable mixed state which has more than one reduced steerable states. Let us consider the following example:
\begin{align}\label{mon90}
&|\phi_{b}\rangle = \frac{4}{9}(1 + \epsilon) {|\phi^{+}\rangle \langle \phi^{+}|}_{AB} \otimes |0\rangle \langle 0|_{C} +\frac{4}{9}(1 + \epsilon) \times     \nonumber \\
 & {|\phi^{+}\rangle \langle \phi^{+}|}_{AC}  \otimes |0\rangle \langle 0|_{B} + \frac{1}{9}(1 - 8\epsilon) {|\phi^{+}\rangle \langle \phi^{+}|}_{BC} \otimes |0\rangle \langle 0|_{A} \nonumber\\
 \end{align}
where $0 \leq \epsilon \leq 1$ and $|\phi^{+}\rangle = \frac{|00\rangle + |11\rangle}{\sqrt{2}}.$ For this biseparable mixed state, the bipartite reduced states $\rho_{AB}$ and $\rho_{AC}$ are $F_{3}$ steerable if $\epsilon > \frac{9}{4\sqrt{3}} - 1$. Thus, genuine entanglement is not necessary to reveal non-monogamous nature of steering correlations.


\section{Discussions}\label{sec4}
Analysis of shareability of correlations between parties sharing a quantum system is an effective way of interpreting quantum theory. In this paper, we have investigated the shareability properties of quantum steering correlations. For our purpose, we have considered the three settings linear steering ($F_{3}$) inequality. Interestingly it is observed that at most two reduced states of any arbitrary three qubit state can violate the $F_{3}$ inequality. This in turn reveal non monogamous nature of steering correlations. Such an observation is however in contrary to the monogamous nature of steering obtained while using two setting linear steering inequality or Bell-CHSH inequality. This indicates that steering correlations can be non-monogamous depending on the measurement scenario.
Now steering correlations in a setup with two settings per party cannot be shared, whereas the same is possible when a setup with three settings per party is considered. So it might be tempting to think that increase of more settings per party could provide more steerable reduced states. Consequently, it will be interesting if one investigate this shareability phenomenon for more than three settings scenario.\\
We have also addressed the question how different measures of genuine entanglement and also entanglement of reduced states relate with reduced bipartite steering of three qubit states. We have established several relations between reduced bipartite steering and different measures of entanglement. Relation existing between bipartite steering, Bell-CHSH nonlocality, and genuine entanglement for three qubit states has also been analyzed.\\
Then, we have determined the complementarity relation between the local information content and bipartite steering. We believe that this will be helpful in designing some appropriate information-theoretic measures of steering. Moreover, we have shown that the shareability constraints allow us to detect different types of tripartite entanglement. Now, monogamy is the essential part in ensuring the security of quantum cryptographic protocols \cite{Pal}. For this reason, it is beneficial to capture precisely under what condition steering correlations is monogamous. So, our observations may be used in framing some more secured quantum cryptographic protocols.\\	
We hope that our results will be useful for further understanding of formalism underlying steering correlations and  their distribution in multipartite states. Apart from investigating our work for more than three setting scenario, it will be interesting to generalize the shareability concept of steering correlations and relations between different quantum correlations for more than two party reduced states. Also, investigation of the same beyond qubit systems is a source of potential future research.


\onecolumngrid

\newpage
\appendix
\begin{center}
\begin{minipage}[c]{\textwidth}
\section{Reduced Bipartite Steering of Three Qubit States}\label{Classi}
To check the number of reduced steerable states of any pure three qubit state, we consider general Schmidt decomposition (GSD) of three
qubit pure states as follows: \cite{Aac}
\begin{equation}\label{A1}
|\phi\rangle = \lambda_{0} |000\rangle + \lambda_{1} e^{ \imath \phi} |100\rangle + \lambda_{2} |101\rangle + \lambda_{3} |110\rangle +  \lambda_{4} |111\rangle,
 \end{equation}
where $\lambda_{i} \geq 0$, $\sum_{i} \lambda_{i} = 1$ and $\phi$ is a phase between $0$ and $\pi$.
It is direct to derive that \cite{Che}, $\vec{a} = (2\lambda_{0}\lambda_{1}\cos \phi, 2\lambda_{0} \lambda_{1}\sin \phi, 2 \lambda_{0}^{2} - 1)$, $\vec{b} = (2\lambda_{1}\lambda_{3}\cos \phi + 2 \lambda_{2} \lambda_{4}, - 2\lambda_{1}\lambda_{3}\sin \phi, 1 - 2 \lambda_{3}^{2} - 2 \lambda_{4}^{2})$, $\vec{c} = (2\lambda_{1}\lambda_{2}\cos \phi + 2 \lambda_{2} \lambda_{4}, - 2\lambda_{1}\lambda_{2}\sin \phi, 1 - 2 \lambda_{2}^{2} - 2 \lambda_{4}^{2})$. From the formulae of calculating $S_{ij}$ presented in Eqs.(\ref{mon15},\ref{mon16},\ref{mon17}), one can provide the following expressions of $S_{ij}$ for any three qubit state in $|\phi\rangle$ :
\begin{equation}\label{A2}
S_{AB} = 1 + 8 \lambda_{0}^{2} \lambda_{3}^{2} - 4 \lambda_{0}^{2} \lambda_{2}^{2} - 4 \lambda_{1}^{2} \lambda_{4}^{2} - 4 \lambda_{2}^{2} \lambda_{3}^{2} + 8 \lambda_{1} \lambda_{2} \lambda_{3} \lambda_{4} \cos\phi,
\end{equation}

\begin{equation}\label{A3}
S_{AC} = 1 + 8 \lambda_{0}^{2} \lambda_{2}^{2} - 4 \lambda_{0}^{2} \lambda_{3}^{2} - 4 \lambda_{1}^{2} \lambda_{4}^{2} - 4 \lambda_{2}^{2} \lambda_{3}^{2} + 8 \lambda_{1} \lambda_{2} \lambda_{3} \lambda_{4} \cos\phi,
\end{equation}
and
\begin{equation}\label{A4}
S_{BC} = 1 - 4 \lambda_{0} \lambda_{2}^{2} - 4 \lambda_{0}^{2} \lambda_{3}^{2} + 8 \lambda_{1}^{2} \lambda_{4}^{2} + 8 \lambda_{2}^{2} \lambda_{3}^{2} - 16 \lambda_{1} \lambda_{2}\lambda_{3} \lambda_{4} \cos\phi.
\end{equation}
By somewhat tedious but straightforward calculations, we obtain concurrence of each bipartite reduced state \cite{Wot}: $\mathcal{C}_{AB}^{2} = 4 \lambda_{0}^{2} \lambda_{3}^{2}$, $\mathcal{C}_{AC}^{2} = 4 \lambda_{0}^{2} \lambda_{2}^{2}$, and $\mathcal{C}_{BC}^{2} = 4 \lambda_{2}^{2} \lambda_{3}^{2} +  4 \lambda_{1}^{2} \lambda_{4}^{2} - 8 \lambda_{1} \lambda_{2} \lambda_{3} \lambda_{4} \cos \phi$.
In \cite{Sab}, Sab\'{i}n and Garc\'{i}a-Alcaine have proposed a classification of three-qubit states based on the existence of bipartite and tripartite entanglements. Here we investigate the number of reduced steerable states in each of those classes of states. Different types of reduced steering are summarised in Fig.(\ref{monp}).\\
(i)\textbf{Type- 0-0 (Fully separable state)} : A pure state $|\phi\rangle$ is fully separable if it can be written
as $|\phi_{1}\rangle \otimes |\phi_{2}\rangle \otimes |\phi_{3}\rangle$. Clearly all reduced states are separable, thereby implying no $F_{3}$ steerable reduced states. The corresponding steering graph(case (a) in Fig(\ref{monp})) has three vertices without any edge. \\
(ii) \textbf{Subtype- $1^{1} -1$ (Biseparable state) }: any state of this class has one of the following GSD forms:\\
$|\phi_{BS}\rangle =   \lambda_{1} e^{ \imath \phi} |100\rangle + \lambda_{2} |101\rangle + \lambda_{3} |110\rangle +  \lambda_{4} |111\rangle$, where $ \lambda_{1}  \lambda_{4} \neq  \lambda_{2} \lambda_{3}$ and  $\lambda_{1}  \lambda_{4}$ or $ \lambda_{2} \lambda_{3}$ can be zero, if $ \lambda_{1}  \lambda_{4} =  \lambda_{2} \lambda_{3}$ the state is of type $0-0$.\\
$|{\phi'}_{BS}\rangle =   \lambda_{0} |000\rangle + \lambda_{1} e^{ \imath \phi} |100\rangle + \lambda_{2} |101\rangle $,\\
$|{\phi''}_{BS}\rangle =  \lambda_{0} |000\rangle +  \lambda_{1} e^{ \imath \phi} |100\rangle  + \lambda_{3} |110\rangle$ where $\lambda_{1}$ can be zero in the last two cases.\\
In each case, exactly one of the reduced states is $F_{3}$ steerable. For example, $\rho_{BC}$ is the only reduced $F_{3}$ steerable state of $|\phi_{BS}\rangle$. So, any biseparable pure state will obey monogamy of steering. The corresponding steering graph has only one edge connecting two circles (see Fig.\ref{monp}(b)). \\
(iii) \textbf{Subtype-$2-0$ (GHZ like states)}:
This class of states has the form : $|\phi_{GGHZ}\rangle = \alpha |000\rangle + \beta |111\rangle$, where $\alpha^{2} + \beta^{2} = 1.$ It includes the $|\phi_{GHZ}\rangle =\frac{1}{\sqrt{2}}( |000\rangle + |111\rangle)$  state. Entanglement of this class of states cannot be persisted if one of the qubit is traced out. Hence, none of the reduced state can be $F_{3}$ steerable. Three circle without any edge (see Fig.\ref{monp}(c)) corresponds this classes of states. Thus, we see that two types of states (GHZ-like states
and separable states) have the same graph.\\
(iv) \textbf{Subtype- $2-1$ (Extended GHZ states)} : any state of this class is one the following GSD forms:\\
$|\phi_{EGHZ}\rangle = \lambda_{0} |000\rangle + \lambda_{1} e^{ \imath \phi} |100\rangle + \lambda_{4} |111\rangle$,\\
$|{\phi'}_{EGHZ}\rangle = \lambda_{0} |000\rangle + \lambda_{2} |101\rangle + \lambda_{4} |111\rangle$,\\
$|{\phi''}_{EGHZ}\rangle = \lambda_{0} |000\rangle + \lambda_{3} |110\rangle + \lambda_{4} |111\rangle$,
with the three nonzero coefficients  in each case.\\ Any state of this class has only one entangled reduced state. For example, the entangled reduced state $\rho_{AC}$ of $|\phi_{EGHZ} \rangle $ is given by $\rho_{AC} =  |\alpha\rangle \langle \alpha | + \lambda_{4}^{2} |11\rangle \langle 11|$ where $|\alpha \rangle =
\lambda_{0} |00\rangle + \lambda_{2}|11\rangle$, with concurrence $\mathcal{C}_{AC}^{2} = 4 \lambda_{0}^{2} \lambda_{2}^{2}$.  Since $\mathcal{C}_{AB}^{2}$ and $\mathcal{C}_{BC}^{2}$ both are equal to zero, $S_{AC}$ can be obtained straightforwardly from Eq.(\ref{mon99}) and its permutations as $S_{AC} = 1 + 2\, \mathcal{C}_{AC}^{2}.$ Thus, any extended GHZ state has only one reduced $F_{3}$ steerable states and thereby maintaining monogamous nature. Hence, biseparable states and extended GHZ states have the same graph.      \\
(v) \textbf{Subtype- $2-2$ (Star shaped states)}: This class of states takes one of the following GSD forms : \\
$|\phi_{STAR}\rangle = \lambda_{0} |000\rangle + \lambda_{1} e^{ \imath \phi} |100\rangle + \lambda_{2} |101\rangle +  \lambda_{4} |111\rangle$,\\
$|{\phi'}_{STAR}\rangle = \lambda_{0} |000\rangle + \lambda_{1} e^{ \imath \phi} |100\rangle + \lambda_{3} |110\rangle +  \lambda_{4} |111\rangle$ with all coefficients nonzero. This class of states belongs to the class of GHZ \cite{Ple}, since it contains genuine entanglement with $\tau = 4 \lambda_{0}^{2} \lambda_{4}^{2}.$ It is the only class of states among all GHZ class of states that can have two entangled reduced states. We find that for $|{\phi'}_{STAR}\rangle$, $\mathcal{C}_{AC}$ is always zero while $\mathcal{C}_{AB}^{2} (= 4 \lambda_{0}^{2} \lambda_{3}^{2}$) and $\mathcal{C}_{BC}^{2} (= 4 \lambda_{1}^{2} \lambda_{4}^{2}$) are nonzero. Combining these with Eq.(\ref{mon99}) and its permutations, one can find that state belongs to this class will obey non-monogamy
if and only if $4 \lambda_{1}^{2} \lambda_{4}^{2} > 2 \lambda_{0}^{2} \lambda_{2}^{2} > \lambda_{1}^{2} \lambda_{4}^{2}$ holds.

\end{minipage}
\end{center}
One simple example of such state is : $\sqrt{\frac{11}{64}}|000\rangle + \sqrt{\frac{5}{64}} |100\rangle + \frac{1}{2} |110\rangle + \frac{1}{\sqrt{2}} |111\rangle. $ Similarly, one can also find a state from this class which violates the above mentioned inequality : $\sqrt{\frac{3}{32}}|000\rangle + \sqrt{\frac{5}{32}} |100\rangle + \frac{1}{2} |110\rangle + \frac{1}{\sqrt{2}} |111\rangle$.
Thus, this class of states can be both monogamous and non-monogamous. Also it is clear that any state of this class has atleast one steerable reduced state, since $S_{AB} + S_{BC} = 2 + 4 \lambda_{0}^{2} \lambda_{3}^{2} + 4 \lambda_{1}^{2}\lambda_{4}^{2} >2$ for every nonzero value of state parameters. This class of states is represented by Fig.\ref{monp}(b) and \ref{monp}(c).\\
(vi) \textbf{Subtype- $2-3$($W$- like states)}: We now take $W$-like states into account. This class of states is given by the following :
$|\phi_{W}\rangle = \lambda_{0} |000\rangle + \lambda_{1} e^{ \imath \phi} |100\rangle + \lambda_{2} |101\rangle + \lambda_{3} |110\rangle$, where $\lambda_{0}, \lambda_{2},\lambda_{3} >0 $ and $\lambda_{1}^{2} \geq 0$. \\
For $W$-like states, all bipartite entanglements are non-zero, with $\mathcal{C}_{AB}^{2} = 4 \lambda_{0}^{2} \lambda_{3}^{2}$ , $\mathcal{C}_{AC}^{2} = 4 \lambda_{0}^{2} \lambda_{2}^{2}$, and $\mathcal{C}_{BC}^{2} = 4 \lambda_{2}^{2} \lambda_{3}^{2}$. At this point one might wonder whether $W$-class contains states with no reduced steering. Let us consider the $|W\rangle = \frac{1}{\sqrt{3}}(|001\rangle + |010\rangle + |100\rangle)$ state. As shown in sec.(\ref{monsec1}), it has no reduced steering. This is in contrast to GHZ state $|\phi_{GHZ}\rangle$ which are less bipartite entangled but have same steering graph. Let us now come to the question of monogamy(or non-monogamy) for states in the $W$ class. From the criterion presented in corollary (\ref{monc2}), monogamy holds for this class of states if and only if $ H(\lambda_{i}^{2}, \lambda_{j}^{2}) < \lambda_{k}^{2}$, $(i \neq j \neq k, i, j, k = 0,2,3)$ for any two sets of values of $(i, j, k)$,  where $H(\lambda_{i}^{2}, \lambda_{j}^{2})$ denotes the harmonic mean of $\lambda_{i}^{2} $ and $\lambda_{j}^{2}$. This is the the only class of states where one can get all types of steering graphs i.e., no reduce steering states, one reduced steering states, and also two reduced steering states. Examples of one reduced steering state and two reduced steering state are given below : one reduced steering state: $ \frac{1}{\sqrt{6}}( |000\rangle + | 100\rangle + |101\rangle ) + \frac{1}{\sqrt{2}} |110\rangle$, two reduced steering : $|\psi_{ABC}\rangle = \frac{1}{2}(|100\rangle +|010\rangle + \sqrt{2} |001\rangle).$ From the above analysis, it is clear that this class of states can correspond to any steering graph (Fig.\ref{monp}(a)-(c)).  \\

From the above classification, we want to remark that only star shaped states(subtype-$2-2$) and $W$-like(subtype-$2-3$) states can violate monogamy of steering correlations. We believe that our classification of three qubit pure states in terms of reduced steering and monogamy (or non-monogamy) can be useful in many areas of quantum information.

Now we investigate the effect of admixing white noise to those two classes of pure states(star shape states and $W$-like states
) which can exhibit non-monogamous nature of steering correlations. In order to analyse it, we define a critical value $v$  ($0\leq v \leq 1$)
for which the mixed states defined by
\begin{equation}\label{A5}
\rho_{star} = v(|{\phi}'_{star}\rangle \langle {\phi}'_{star} |) + (1-v)\frac{\mathbf{I}}{8}
\end{equation}
and
\begin{equation}\label{A6}
\rho_{W} = v(|\phi_{W}\rangle \langle \phi_{W}|) + (1-v)\frac{\mathbf{I}}{8}
\end{equation}
looses the non-monogamous nature of the original pure states. For a given noisy state, we intend to find the critical value $v_{crit}$ such that if $v>v_{crit}$ the non-monogamy nature is preserved for steering correlations i.e., there exist two steerable reduced states of the given noisy state.

For the $\rho_{star}$ state, one has $S_{AB} = v(1 + 8 \lambda_{0}^{2} \lambda_{2}^{2} - 4 \lambda_{1}^{2} \lambda_{4}^{2})$, $S_{BC} = v(1 + 8 \lambda_{1}^{2} \lambda_{4}^{2} - 4 \lambda_{0}^{2} \lambda_{2}^{2})$, and $S_{AC} = 1 - 4 \lambda_{1}^{2} \lambda_{4}^{2} - 4 \lambda_{0}^{2} \lambda_{2}^{2}$. Consequently, the state $\rho_{star}$ leads to the critical visibility
\begin{equation}\label{A7}
v_{crit} = \max[\frac{1}{1 + 8 \lambda_{0}^{2} \lambda_{2}^{2} - 4 \lambda_{1}^{2} \lambda_{4}^{2}}, \frac{1}{1 + 8 \lambda_{1}^{2} \lambda_{4}^{2} - 4 \lambda_{0}^{2} \lambda_{2}^{2}}].
\end{equation}
Notice that $v_{crit}$ is minimised for $\lambda_{i}^{2} = \frac{1}{4}$ $(i= 0,1, 2, 3)$ which corresponds to the state $|\phi^{star}_{v}\rangle = \frac{1}{2} |000\rangle + \frac{1}{2} |100\rangle + \frac{1}{2} |110\rangle + \frac{1}{2} |111\rangle$ and leads to $v_{crit} = 0.8$. Thus, the state $ |\phi^{star}_{v}\rangle$ is more robust against white noise than any other non-monogamous $|{\phi}'_{star}\rangle$ states.\\
Similarly, one can find
\begin{equation}\label{A8}
v_{crit} = \min[ \max \{w_{1},w_{2}\}, \max\{w_{1}, w_{3}\}, \max\{w_{2},w_{3}\}]
\end{equation}
for $\rho_{W}$ state, where $w_{1} = \frac{1}{1 + 8 \lambda_{0}^{2} \lambda_{3}^{2} - 4 \lambda_{0}^{2} \lambda_{2}^{2} - 4 \lambda_{2}^{2} \lambda_{3}^{2}}$, $w_{2} = \frac{1}{1 + 8 \lambda_{0}^{2} \lambda_{2}^{2} - 4 \lambda_{0}^{2} \lambda_{3}^{2} - 4 \lambda_{2}^{2} \lambda_{3}^{2}}$ and $w_{3} = \frac{1}{1 + 8 \lambda_{2}^{2} \lambda_{3}^{2} - 4 \lambda_{0}^{2} \lambda_{2}^{2} - 4 \lambda_{0}^{2} \lambda_{3}^{2}}$. The most robust non-monogamy property is observed for the state $|\phi^{W}_{v}\rangle = \sqrt{\frac{2}{3}}|000\rangle + \sqrt{\frac{1}{6}} |101\rangle + \sqrt{\frac{1}{6}} |110\rangle$ and the corresponding $v_{crit} = 0.75$. Intuitively it can be expected that
higher entangled states might have greater robustness of non-monogamy compared to the lesser entangled states. Let us take the example of $|\psi_{ABC}\rangle$ state which has $v_{crit} = 0.8$. Now if we compare the efficiency of $|\phi^{W}_{v}\rangle$ with $|\psi_{ABC}\rangle$ we find that the less entangled state $|\phi^{W}_{v}\rangle$ with $E_{W} = \frac{1}{9}$ is more robust in comparison to the higher entangled state $|\psi_{ABC}\rangle$ having $E_{W} = \frac{1}{4}$.




\begin{thebibliography}{99}
	
\bibitem{Kim} H. J. Kimble, ``The quantum internet'' , \href{https://doi.org/10.1038/nature07127} {Nature 453, 1023 (2008)}.




\bibitem{Hor} R. Horodecki, P. Horodecki, M. Horodecki, and K. Horodecki, "Quantum entanglement", \href{https://doi.org/10.1103/RevModPhys.81.865}{Rev. Mod. Phys. {\bf 81}, 865 (2009)}. 	

\bibitem{Guh} O. G\"{u}hne, G. T\'{o}th, ``Entanglement detection'', \href{http://www.sciencedirect.com/science/article/pii/S0370157309000623}{Physics Reports 474, 1 (2009)}.

\bibitem{Bru} N. Brunner, D. Cavalcanti, S. Pironio, V. Scarani, and S.Wehner, "Bell nonlocality", \href{http://journals.aps.org/rmp/abstract/10.1103/RevModPhys.86.419}{Rev. Mod. Phys. {\bf 86}, 839 (2014)}. 	

\bibitem{Ton} B. Toner, "Monogamy of non-local quantum correlations", \href{https://doi.org/10.1098/rspa.2008.0149}{Proc. R. Soc. A  {\bf 465}, 59 (2008)}. 	

\bibitem{Cof} V. Coffman, , J. Kundu, W. K. Wootters, "Distributed entanglement", \href{https://doi.org/10.1103/PhysRevA.61.052306}{Phys. Rev. A  {\bf 61}, 052306 (2000)}. 	

\bibitem{Osb} T. Osborne, F. Verstraete, "General Monogamy Inequality for Bipartite Qubit Entanglement", \href{https://doi.org/10.1103/PhysRevLett.96.220503}{Phys. Rev. Lett.  {\bf 96}, 220503 (2006)}.

\bibitem{Yco} Y. C. Ou, H. Fan, and S. M. Fei, "Proper monogamy inequality for arbitrary pure quantum states", \href{https://doi.org/10.1103/PhysRevA.78.012311}{Phys. Rev. A  {\bf 78}, 012311 (2008)}.

\bibitem{Bai}  Y. K. Bai, Y. F. Xu, and Z. D. Wang, "General Monogamy Relation for the Entanglement of Formation in Multiqubit Systems", \href{https://doi.org/10.1103/PhysRevLett.113.100503}{Phys. Rev. Lett. {\bf 113}, 100503 (2014)}.

\bibitem{Zhu}  X. N. Zhu and S. M. Fei, "Entanglement monogamy relations of qubit systems", \href{https://doi.org/10.1103/PhysRevA.90.024304}{Phys. Rev. A {\bf 90}, 024304 (2014)}.

\bibitem{Reg}  B. Regula, S. Di Martino, S. Lee, and G. Adesso, "Strong Monogamy Conjecture for Multiqubit Entanglement: The Four-
Qubit Case", \href{https://doi.org/10.1103/PhysRevLett.113.110501}{Phys. Rev. Lett. {\bf 113}, 110501 (2014)}.

\bibitem{Teh}  B.M. Terhal, "Is entanglement monogamous?", \href{https://doi.org/10.1147/rd.481.0071}{IBM Journal of Research and Development {\bf 48}, 1 (2004)}.

\bibitem{Pal}  M. Pawlowski, "Security proof for cryptographic protocols based only on the monogamy of Bell's inequality violations", \href{https://doi.org/10.1103/PhysRevA.82.032313}{Phys. Rev. A  {\bf 82}, 032313 (2010)}.

\bibitem{Gis}  N. Gisin, G. Ribordy, W. Tittel, H. Zbinden, "Quantum cryptography", \href{https://doi.org/10.1103/RevModPhys.74.145}{Rev. Mod. Phys. {\bf 74}, 145 (2002) }.

\bibitem{Dur}  W. D\"{u}r, G. Vidal, J.I. Cirac, "Three qubits can be entangled in two inequivalent ways", \href{https://doi.org/10.1103/PhysRevA.62.062314}{Phys. Rev. A {\bf 62}, 062314 (2000) }.
	
\bibitem{Gio}  G.L. Giorgi, "Monogamy properties of quantum and classical correlations", \href{https://doi.org/10.1103/PhysRevA.84.054301}{Phys. Rev. A {\bf 84}, 054301 (2011) }.
	
	
\bibitem{Pra}  R. Prabhu, A.K. Pati, A. Sen(De), U. Sen, "Conditions for monogamy of quantum correlations: Greenberger-Horne-Zeilinger versus W states", \href{https://doi.org/10.1103/PhysRevA.85.040102}{Phys. Rev. A {\bf 85}, 040102(R) (2012) }.
	

\bibitem{Llo}  S. Lloyd, J. Preskill, "Unitarity of black hole evaporation in final-state projection models", \href{https://doi.org/10.1007/JHEP08(2014)126}{J. High Energy Phys. {\bf 85}, 040102(R) (2012) }.

\bibitem{Rao}  K.R.K. Rao, H. Katiyar, T.S. Mahesh, A. Sen(De), U. Sen, A. Kumar, "Multipartite quantum correlations reveal frustration in a quantum Ising spin system", \href{https://doi.org/10.1103/PhysRevA.88.022312}{Phys. Rev. A {\bf 88}, 022312 (2013) }.

\bibitem{Cla} J. F. Clauser, M. A. Horne, A. Shimony, and R. A. Holt, ``Proposed Experiment to Test Local Hidden-Variable Theories'', \href{https://doi.org/10.1103/PhysRevLett.23.880}{Phys. Rev. Lett. {\bf 23}, 880 (1969)}.




\bibitem{Sca}  V. Scarani and N. Gisin, "Quantum Communication between N Partners and Bell's Inequalities", \href{https://doi.org/10.1103/PhysRevLett.87.117901}{Phys. Rev. Lett. {\bf 87}, 117901 (2001) }.

\bibitem{Scn}  V. Scarani and N. Gisin, "Quantum key distribution between N partners: Optimal eavesdropping and Bell's inequalities", \href{https://doi.org/10.1103/PhysRevA.65.012311}{Phys. Rev. A {\bf 65}, 012311 (2001) }.


\bibitem{Tne}  B. Toner, F. Verstraete, "Monogamy of Bell correlations and Tsirelson's bound", \href{https://arxiv.org/abs/0611001}{Arxiv: 0611001}.

\bibitem{Bar}  J. Barrett, N. Linden, S. Massar, S. Pironio, S. Popescu,
and D. Roberts, "Nonlocal correlations as an information-theoretic resource", \href{https://doi.org/10.1103/PhysRevA.71.022101}{Phys. Rev. A {\bf 71}, 022101 (2005) }.

\bibitem{Mas}  Ll. Masanes, A. Acin, and N. Gisin, "General properties of nonsignaling theories", \href{https://doi.org/10.1103/PhysRevA.73.012112}{Phys. Rev. A {\bf 73}, 012112 (2006) }.
\bibitem{See}  M. P. Seevinck, "Monogamy of correlations versus monogamy of entanglement", \href{https://doi.org/10.1007/s11128-009-0161-6}{Quantum Inf Process {\bf 9}, 273-294 (2010) }.

\bibitem{Paw}  M. Pawlowski and C. Brukner, "Monogamy of Bell's Inequality Violations in Nonsignaling Theories", \href{https://doi.org/10.1103/PhysRevLett.102.030403}{Phys. Rev. Lett. {\bf 102}, 030403 (2009) }.

\bibitem{Kur}  P. Kurzynski, T. Paterek, R. Ramanathan, W.
Laskowski, D. Kaszlikowski, "Correlation Complementarity Yields Bell Monogamy Relations", \href{https://doi.org/10.1103/PhysRevLett.106.180402}{Phys. Rev. Lett. {\bf 106}, 180402 (2011) }.
	
\bibitem{Qin}  H. H. Qin, S. M. Fei, and X. Li-Jost, "Trade-off Relations of Bell Violations among Pairwise Qubit
Systems", \href{https://doi.org/10.1103/PhysRevA.92.062339}{Phys. Rev. A {\bf 92}, 062339 (2015) }.
	
\bibitem{Che}  S. Cheng, M.J.W. Hall, "Anisotropic Invariance and the Distribution of Quantum Correlations", \href{https://doi.org/10.1103/PhysRevLett.118.010401}{Phys. Rev. Lett. {\bf 118}, 010401 (2017) }.

\bibitem{Tra}  M. C. Tran, R. Ramanathan, M. McKague, D. Kaszlikowski, and T. Paterek, "Bell monogamy relations in arbitrary qubit networks", \href{https://doi.org/10.1103/PhysRevA.98.052325}{Phys. Rev. A {\bf 98}, 052325 (2018) }.

\bibitem{Ram}  R. Ramanathan1 and P. Mironowicz, "Trade-offs in multiparty Bell-inequality violations in qubit networks", \href{https://doi.org/10.1103/PhysRevA.98.022133}{Phys. Rev. A {\bf 98}, 022133 (2018) }.


\bibitem{Ste}  A. Streltsov, G. Adesso, M. Piani, and D.Bruss, "Are General Quantum Correlations Monogamous?", \href{https://doi.org/10.1103/PhysRevLett.109.050503}{Phys. Rev.
Lett. {\bf 109}, 050503 (2012) }.

\bibitem{Kar}  M. Karczewski, D. Kaszlikowski, and P. Kurzy\'{n}ski, "Monogamy of Particle Statistics in Tripartite Systems Simulating Bosons and Fermions", \href{https://doi.org/10.1103/PhysRevLett.121.090403}{Phys. Rev.
Lett. {\bf 121}, 090403 (2018)}.

\bibitem{Rad}  C. Radhakrishnan, M. Parthasarathy, S. Jambulingam,
and T. Byrnes, "Distribution of Quantum Coherence in Multipartite Systems", \href{https://doi.org/10.1103/PhysRevLett.116.150504}{Phys. Rev. Lett. {\bf 116}, 150504 (2016)}.

\bibitem{Chn}  S. Cheng and L. Liu, "Monogamy relations of nonclassical correlations for multi-qubit states", \href{https://doi.org/10.1016/j.physleta.2018.04.037}{Phys. Lett. A  {\bf 26}, 1716 (2018) }.

\bibitem{Rei}  M. D. Reid, "Monogamy inequalities for the Einstein-Podolsky-Rosen paradox and quantum steering", \href{https://doi.org/110.1103/PhysRevA.88.062108}{Phys. Rev. A  {\bf 88}, 062108 (2013) }.

\bibitem{Xia}  Y. Xiang, I. Kogias, G. Adesso, and Q. Y. He, "Multipartite Gaussian steering: Monogamy constraints and quantum cryptographyapplications", \href{https://doi.org/10.1103/PhysRevA.95.010101}{Phys. Rev. A  {\bf 95}, 010101(R) (2017) }.

\bibitem{Chg}  S. Cheng, A. Milne, M. J. W. Hall, and H. M. Wiseman, "Volume
monogamy of quantum steering ellipsoids for multiqubit
systems", \href{https://doi.org/10.1103/PhysRevA.94.042105}{Phys. Rev. A  {\bf 94}, 042105 (2016) }.



\bibitem{Sch}  E. Schr\"{o}dinger, "Discussion of probability relations between
separated systems", \href{https://doi.org/10.1017/S0305004100013554}{Proc. Cambridge Philos. Soc.  {\bf 31}, 555-563 (1935) }.

\bibitem{Wis} H. M. Wiseman, S. J. Jones, and A. C. Doherty, ``Steering, Entanglement, Nonlocality, and the Einstein-Podolsky-Rosen Paradox'', \href{http://journals.aps.org/prl/abstract/10.1103/PhysRevLett.98.140402}{Phys. Rev. Lett. 98, 140402 (2007)}.


\bibitem{Jon} S. J. Jones, H. M. Wiseman, and A. C. Doherty, ``Entanglement, Einstein-Podolsky-Rosen correlations, Bell nonlocality, and steering'', \href{http://journals.aps.org/pra/abstract/10.1103/PhysRevA.76.052116}{Phys. Rev. A 76, 052116 (2007)}.


\bibitem{Uol} R. Uola, A. C. S. Costa, H. Chau Nguyen, and O. G\"{u}hne, ``Quantum steering'', \href{http://journals.aps.org/pra/abstract/10.1103/RevModPhys.92.015001}{Rev. Mod. Phys. 92, 015001 (2020)}.


\bibitem{Qui}  M. T. Quintino, T. V\'{e}rtesi, D. Cavalcanti, R. Augusiak, M. Demianowicz, A. Ac\'{i}n, and N. Brunner, ``Inequivalence of entanglement, steering, and Bell nonlocality for general measurements'', \href{http://dx.doi.org/10.1103/PhysRevA.92.032107}{Phys. Rev. A {\bf 92}, 032107 (2015)}.


\bibitem{Red}  M. D. Reid, A. Ac\'{i}n, and N. Brunner, ``Demonstration of the Einstein-Podolsky-Rosen
paradox using nondegenerate parametric amplification'', \href{http://dx.doi.org/10.1103/PhysRevA.40.913}{Phys. Rev. A {\bf 40}, 913 (1989)}.

\bibitem{Cav}  E. G. Cavalcanti, S. J. Jones, H. M. Wiseman, and M. D. Reid, ``Experimental criteria for steering and the Einstein-Podolsky-Rosen paradox'', \href{http://dx.doi.org/10.1103/PhysRevA.80.032112}{Phys. Rev. A {\bf 80}, 032112 (2009)}.

\bibitem{Wal} S. P. Walborn, A. Salles, R. M. Gomes, F. Toscano, and P.
H. Souto Ribeiro, ``Revealing Hidden Einstein-Podolsky-Rosen
Nonlocality'', \href{http://journals.aps.org/prl/abstract/10.1103/PhysRevLett.106.130402}{Phys. Rev. Lett. {\bf 106}, 130402 (2011)},

\bibitem{Sco} J. Schneeloch, C. J. Broadbent, S. P. Walborn, E. G. Cavalcanti,
and J. C. Howell, ``Einstein-Podolsky-Rosen steering inequalities
from entropic uncertainty relations'', \href{http://journals.aps.org/prl/abstract/10.1103/PhysRevA.87.062103}{Phys. Rev. A {\bf 87}, 062103 (2013)},

\bibitem{Jlc} J.-L. Chen, X.-J. Ye, C. Wu, H.-Y. Su, A. Cabello, L. C.
Kwek,and C. H. Oh, ``All-versus-nothing proof of Einstein-
Podolsky-Rosen steering'', \href{http://journals.aps.org/prl/abstract/10.1038/srep02143}{Sci. Rep. {\bf 3}, 2143 (2013)},

\bibitem{Kog} I. Kogias, P. Skrzypczyk, D. Cavalcanti, A. Ac\'{i}n, and G.
Adesso, ``Hierarchy of Steering Criteria Based on Moments for
All Bipartite Quantum Systems', \href{http://journals.aps.org/prl/abstract/10.1103/PhysRevLett.115.210401}{Phys. Rev. Lett. {\bf 115}, 210401 (2015)},

\bibitem{Cva} IE. G. Cavalcanti, C. J. Foster, M. Fuwa, and H. M. Wiseman, ``Analog of the Clauser-Horne-Shimony-Holt inequality for
steering', \href{http://journals.aps.org/prl/abstract/10.1364/JOSAB.32.000A40}{J. Opt. Soc. Am. B {\bf 32}, A74 (2015)},


\bibitem{Zuk} M. Zukowski, A. Dutta, and Z. Yin, ``Geometric Bell-like inequalities
for steering', \href{http://journals.aps.org/prl/abstract/10.1103/PhysRevA.91.032107}{Phys. Rev. A {\bf 91}, 032107 (2015)},

\bibitem{Jev} S. Jevtic, M. J. W. Hall, M. R. Anderson, M. Zwierz, and H. M. Wiseman, ``Einstein–Podolsky–Rosen steering and the steering ellipsoid'', \href{https://www.osapublishing.org/josab/abstract.cfm?uri=josab-32-4-A40}{Journal of the Optical Society of America B, {\bf 32}, A40 (2015)}.

\bibitem{Cos} A. C. S. Costa and R. M. Angelo, ``Quantification of Einstein-Podolski-Rosen steering for two-qubit states'', \href{https://doi.org/10.1103/PhysRevA.93.020103}{Phys. Rev. A, {\bf 93}, 020103 (2016)}.

\bibitem{Law} Y. Z. Law, L. P. Thinh, J.-D. Bancal, and V. Scarani, ``Quantum
randomness extraction for various levels of characterization of
the devices'', \href{https://doi.org/10.1088/1751-8121/aad115}{J. Phys. A: Math. Theor. {\bf 47}, 424028 (2016)}.

\bibitem{Pia} M. Piani, and J.Watrous, ``Necessary and Sufficient Quantum Information
Characterization of Einstein-Podolsky-Rosen Steering'', \href{https://doi.org/10.1103/PhysRevLett.114.060404}{Phys. Rev. Lett. {\bf 114}, 060404 (2015)}.

\bibitem{Bra} C. Branciard and N. Gisin, ``Quantifying the Nonlocality
of Greenberger-Horne-Zeilinger Quantum Correlations by a
Bounded Communication Simulation Protocol'', \href{https://doi.org/10.1103/PhysRevLett.107.020401}{Phys. Rev. Lett. {\bf 107}, 020401 (2011)}.

\bibitem{Ban} C. Branciard, E. G. Cavalcanti, S. P. Walborn, V. Scarani, and
H. M. Wiseman, ``One-sided device-independent quantum key
distribution: Security, feasibility, and the connection with steering'', \href{https://doi.org/10.1103/PhysRevA.85.010301}{Phys. Rev. A {\bf 85}, 010301(R) (2012)}.



\bibitem{Bne} N. Brunner, J. Sharam, and T. Vertesi, ``Testing the Structure of Multipartite Entanglement with Bell Inequalities'', \href{https://doi.org/10.1103/PhysRevLett.108.110501}{Phys. Rev. Lett. {\bf 108}, 110501 (2012)}.

\bibitem{Wni} M. Wiesniak, M. Nawareg, and M. Zukowski, ``N-particle nonclassicality without N-particle correlations'', \href{https://doi.org/10.1103/PhysRevA.86.042339}{Phys. Rev. A {\bf 86}, 042339 (2012)}.

\bibitem{Tur} J. Tura, R. Augusiak, A. B. Sainz et al., ``Detecting nonlocality in many-body quantum states'', \href{https://doi.org/10.1126/science.1247715}{Science  {\bf 344}, 6189 (2012)}.

\bibitem{Tot} G. T\'{o}th, ``Entanglement witnesses in spin models'', \href{https://doi.org/10.1103/PhysRevA.71.010301}{Phys. Rev. A  {\bf 71}, 010301(R) (2005)}.

\bibitem{Kor} J. Korbickz, J. I. Cirac, J. Wehr et al., ``Hilbert's 17th Problem and the Quantumness of States'', \href{https://doi.org/10.1103/PhysRevA.71.010301}{Phys. Rev. Lett.  {\bf 94}, 153601 (2005)}.

\bibitem{Kna} G. T\'{o}th, C. Knapp, and O. G\"{u}nhe et al., ``Optimal Spin Squeezing Inequalities Detect Bound Entanglement in Spin Models'', \href{https://doi.org/10.1103/PhysRevA.71.010301}{Phys. Rev. Lett.  {\bf 99}, 250405 (2007)}.

\bibitem{Las} M. Markiewicz, W. Laskowski, T. Paterek et al., ``Detecting genuine multipartite entanglement of pure states with bipartite correlations'', \href{https://doi.org/10.1103/PhysRevA.87.034301}{Phys. Rev. Lett.  {\bf 87}, 034301 (2013)}.



\bibitem{Luo} S. Luo, ``Quantum discord for two-qubit systems'', \href{https://doi.org/10.1103/PhysRevA.85.010301}{Phys. Rev. A {\bf 77}, 042303 (2008)}.

\bibitem{Ken} Ng K. Feng, ``Bell monogamy'', Thesis,  \href{https://hdl.handle.net/10356/63200}{Nanyang Technological University, Singapore, (2015) }.

\bibitem{Ghz} D. M. Greenberger, M. A. Horne, A. Shimony and A.
Zeilinger, ``Bell's theorem without inequalities'', \href{https://doi.org/10.1119/1.16243}{Am. J. Phys. {\bf 58}, 1131 (1990)}.

\bibitem{Pat} P. Agrawal, and A. Pati, ``Perfect teleportation and superdense coding with W states'', \href{https://doi.org/10.1103/PhysRevA.74.062320}{Phys. Rev. A {\bf 74}, 062320 (2006)}.

\bibitem{Sab} C. Sab\'{i}n, and G. Garc\'{i}a-Alcaine, ``A classification of entanglement in three-qubit systems'', \href{https://doi.org/10.1140/epjd/e2008-00112-5}{Eur. Phys. J. D {\bf 48}, 435-442 (2008)}.



\bibitem{Hil} S. Hill and W. K. Wootters, ``Entanglement of a Pair of Quantum Bits'', \href{https://doi.org/10.1103/PhysRevLett.78.5022}{Phys. Rev. Lett. {\bf 78}, 5022(1997)}.

\bibitem{Wot} W. K. Wootters, ``Entanglement of Formation of an Arbitrary State of Two Qubits'', \href{https://doi.org/10.1103/PhysRevLett.80.2245}{Phys. Rev. Lett. {\bf 80}, 2245(1998)}.

\bibitem{Ves} F. Verstraete and M. M. Wolf, ``Entanglement versus Bell Violations and Their Behavior under Local Filtering Operations'', \href{https://doi.org/10.1103/PhysRevLett.89.170401}{Phys. Rev. Lett. {\bf 89}, 170401(2002)}.
	
\bibitem{Nep} R. Nepal, R. Prabhu, A. Sen(De), and U. Sen, ``Maximally-dense-coding-capable quantum states'', \href{https://doi.org/10.1103/PhysRevA.87.032336}{Phys. Rev. A {\bf 87}, 032336(2013)}.

\bibitem{Hro} R. Horodecki, P. Horodecki, M. Horodecki, ``Violating Bell inequality by mixed
states: necessary and sufficient condition'', \href{https://doi.org/10.1016/0375-9601(95)00214-N}{Phys. Lett.
A {\bf 200}, 340 (1995)}.			

\bibitem{Zha} X. Zha, Z. Da, I. Ahmed, D. Zhang and Y. Zhang, ``General monogamy equalities of
complementarity relation and distributive
entanglement for multi-qubit pure states'', \href{https://doi.org/10.1088/1612-202X/aa9ac5}{Laser Phys. Lett. {\bf 15}, 025202(2018)}.
	
\bibitem{Hig} A. Higuchi, A. Sudbery, and J. Szulc, ``One-Qubit Reduced States of a Pure Many-Qubit State: Polygon Inequalities'', \href{https://doi.org/10.1103/PhysRevLett.90.107902}{Phys. Rev. Lett. {\bf 90}, 107902 (2003)}.	
		
\bibitem{Aac} A. Ac\'{i}n, A. Andrianov, L. Costa, E. Jan\'{e}, J. I. Latorre, and R. Tarrach, ``Generalized Schmidt Decomposition and Classification of Three-Quantum-Bit States'', \href{https://doi.org/10.1103/PhysRevLett.85.1560}{Phys. Rev. Lett. {\bf 85}, 1560 (2000)}.	
\bibitem{Ple} M. Plesch and V. Bu\v{z}ek, ``Entangled graphs: Bipartite entanglement in multiqubit systems'', \href{https://doi.org/10.1103/PhysRevA.67.012322}{Phys. Rev. A {\bf 67}, 012322 (2003)}.	 
				
	
\end{thebibliography}
\end{document}